\documentclass[11pt]{article}

\usepackage{amsthm}
\usepackage[utf8]{inputenc}
\usepackage[T1]{fontenc}
\usepackage{fixltx2e}
\usepackage{graphicx}
\usepackage{longtable}
\usepackage{float}
\usepackage{wrapfig}
\usepackage{soul}
\usepackage{textcomp}
\usepackage{marvosym}
\usepackage{wasysym}
\usepackage{latexsym}
\usepackage{amssymb}
\usepackage{todonotes}

\usepackage{hyperref}
\newtheorem{lemma}{Lemma}
\newtheorem{theorem}{Theorem}

\title{ELB-Trees\\An Efficient and Lock-free B-tree Derivative}
\author{Lars F.~Bonnichsen, Sven Karlsson, Christian W.~Probst}
\date{\today}
\hypersetup{
  pdfkeywords={},
  pdfsubject={}}

\begin{document}

\maketitle

\section{Introduction}
\label{sec-1}

This technical report is an extension of the paper of the same title, which is to appear at MUCOCOS'13. The technical report proves correctness of the ELB-trees operations' semantics and that the operations are lock-free.

The following is a brief summary of the design of the datastructure, which is detailed in section 3 of the paper.
All ELB-trees have a permanent root node $r$ with a single child.
ELB-trees are $k$-ary leaf-oriented search tree, or multiway search trees, so internal nodes have up to $k$ children and $k-1$ keys. An ELB-trees contain a set $E_r$ of integer keys in the range $(0;2^{63})$. The key 0 is reserved. Keys have an additional read-only bit: when the read-only bit is set, the key cannot be written to.  ELB-trees offer 3 main operations:
\begin{itemize}
\item Search($e_1$, $e_2$) returns a key $e$ from $E_r$ satisfying $e_1 \le e \le e_2$, if such a key exists. Otherwise it returns $0$.
\item Remove($e_1$, $e_2$) removes and returns a key $e$ from $E_r$ satisfying $e_1 \le e \le e_2$, if such a key exists. Otherwise it returns $0$.
\item Insert($e$) adds $e$ to $E_r$, if $e$ was not in $E_r$ before.
\end{itemize}
ELB-trees can also be used as dictionaries or priority queues by storing values in the least significant bits of the keys.

The operations of ELB-trees cannot generally be expressed as atomic operations, as they occur over a time interval. As a consequence, series of concurrent operations cannot generally be expressed as ocurring serially, that is the semantics are not linearizable.
However, the set $E_r$ is atomic.
$E_r$ is the union of the keys in the leaf nodes of the ELB-tree.
The keys in internal nodes guide tree search.

Section 2 provides formal definitions for terms used throughout the proof.
The proof starts in Section 3 by proving that ELB-trees are leaf-oriented search trees.
We prove through induction, that 
ELB-trees are leaf-oriented search trees initially, and that all operations maintain that property.
The inductive step is assisted by two significant subproofs:
\begin{enumerate}
\item Rebalancing does not change the keys in $E_r$.
\item The keys in leaf nodes are within a permanent range.
\end{enumerate}

These properties hold due to the behavior of rebalancing.
The first subproof shows that rebalancing is deterministic, even when concurrent.
The second shows that leaf nodes have a range of keys they may contain and it never changes.

Given these properties, Section 4 derives the operations' semantics.
Section 5 follows up by proving that the operations are lock-free.
First we prove that some operation has made progress whenever a node is rebalanced.
Next we prove that some operation has made progress whenever any part of an operation is restarted.

Section 6 concludes the technical report with a summary.
\section{Definitions}
\label{sec-2}

This section introduces definitions used in the following proofs of the ELB-trees' properties. The definitions start with the terms used, before moving on to the contents and properties of nodes. Finally the intitial state of ELB-trees is formally defined.

Let $L$ be the set of leaf ndoes, $I$ the set of internal nodes, and $T$ the set of points in time. The sets are disjoint.

Nodes contain:
\begin{description}
\item[$C_i (t)$] list of children of internal node $i$ at time $t$
\item[$S_i (t)$] list of keys in internal node $i$ at time $t$
\item[$E_n (t)$] keys represeted by the node $n$ where at time $t$:\begin{center}$E_n (t) =$ \begin{math} \left\{
     \begin{array}{lr}
       $Non-zero keys in $l & n \in L \\
       \bigcup _{c \in C_i (t)} E_c (t) & : n \in I
     \end{array}
   \right. \end{math}\end{center} 
\end{description}

The following node properties can be derived from their content:
\begin{description}
\item[$D_n (t)$] the descendants of node $n$ at time $t$: \\ $D_n (t)$ = \begin{math} \left\{
     \begin{array}{lr}
       \emptyset & : n \in L\\
       C_n (t) \cup \bigcup _{d \in C_n(t)} D_d (t) & : n \in I
     \end{array}
   \right.\end{math} \\ $n$ is reachable when $reachable_n (t) \equiv n \in (\{r\} \cup D_r (t))$
\item[$parent_n (t)$] the parents of node $n$: \\ $parent_n (t) = \{i \in reachable_r (t) | n \in C_i (t)\}, t \in T$
\end{description}

Initially $r$ has one child $C_r (0) = \langle ic \rangle$, and one grandchild $C_{ic} (0) = \\ \langle ln \rangle$.
The grandchild is an empty leaf node $E_{ln} (0) = \emptyset \wedge E_r (0) = \emptyset$.

\section{Search tree proof}
\label{sec-3}

This section proves that ELB-trees are $k$-ary leaf-oriented search trees.
In such a tree, all nodes except the root have one parent, and all internal nodes have strictly ordered keys.
Specifically the $i$'th key in a node provides an upper bound for the $i$'th child of the node, and a lower bound for the $i + 1$'th child.
The key ordering is formally expressed as: 
\begin{center} $W_i (t) \equiv \forall j \in [0;C_i (t)). E_{{C_i(t)} _t} \subseteq (0; {S_i}_j] \wedge E_{{C_i (t)} _t} \subseteq ({S_i} _j; 2^{63})$ \end{center}
The tree property is formally expressed as:
\begin{center} $\forall n \in reachable_n (t). \left\vert parent_n (t) \right\vert = 1 \vee n = r$ \end{center}
The properties are proven inductively, but doing so requires several intermediate steps.
To begin with, we will show that the behavior of rebalancing of search trees is deterministic, and does not change $E_r$.

\begin{lemma}\label{ro-reb}
Unbalanced nodes and their parent are read-only while rebalancing. \end{lemma}
\begin{proof} While finding the nodes involved in rebalancing, they are made read-only: internal nodes are made read-only by setting their status field, and
leaf nodes are made read-only by setting the read-only bit of all their keys, see Figure~16 in the paper.\end{proof}

\begin{lemma}\label{reb-nodes}
If $W_r$ holds and the unbalanced nodes' parent is still reachable, all threads can find the nodes involved in a rebalancing from the status field of the unbalanced nodes grandparent, . \end{lemma}
\begin{proof} The status field stores the key of the unbalanced node and its parent.
Since $W_r$ holds, the nodes can be found by searching for the key in the grandparent and parent of the unbalanced node.\end{proof}

\begin{lemma}\label{inv-detreb}
Rebalancing completes deterministically exactly once, if $W_r$ holds. \end{lemma}
\begin{proof} Rebalancing finds the involved nodes (Lemma \ref{reb-nodes}) and decides how to rebalance (Lemma \ref{ro-reb}) determinstically. The parent is replaced, and the grandparent's status field is cleared using ABA safe CAS operations, see Section~3b of the paper. The grandparent has the status field \{*,*,*,STEP2\} when replacing the parent, ensuring that the grandparent is reachable when replacing the parent node.\end{proof}

\begin{lemma}\label{Er-reb}
$E_r (t)$ does not change when rebalancing, if $W_r$ holds. \end{lemma}
\begin{proof} The content of balanced nodes and their new parent is copied from the old nodes, while their content is read-only (Lemma \ref{ro-reb}).\end{proof}

The preceding lemmas show that rebalancing is well-behaved in search trees. The following lemmas will show that all operations maintain the tree property and $W_r$.

\begin{lemma}\label{inv-tree}
All operations maintain the tree property, if $W_r$ holds. \end{lemma}
\begin{proof} $descendants_n$ only changes when rebalancing. Specifically, $descendants_n$ changes when replacing an internal node $op$ with a new node $np$. 
The children of $op$ had $op$ as their only parent, so all the children $np$ and $op$ share, will have $np$ as their only parent after rebalancing. The new children have $np$ as their only parent, because they have just been introduced, and the descendants of the new nodes have their parents replaced. Formally: \begin{center} $(\forall c \in C_{op} (t_1). parent_c (t_1) = \{op\}) \Rightarrow \forall c \in C_{np} (t_2). parent_c (t_2) = \{np\}$ \end{center}\end{proof}

\begin{lemma}\label{lrange}
Leaf nodes $l$ have a permanent range $R_l$ of keys they may contain, if $W_r$ holds.\end{lemma}
\begin{proof} The lower bound is given by the keys of its ancestors. The ancestors change deterministically when $W_r$ holds (Lemma \ref{inv-detreb}). Although the ancestors may change, their replacements use the same keys. 
Internal node keys are only introduced or removed when splitting and merging nodes, which results in two or three new nodes. 
When rebalancing results in two new nodes, the new parent has one less key. When rebalancing results in three new nodes, the new parent has one updated or additional key, which the old parent did not have. The updated or new key is copied from its the unbalanced nodes, so it only affects the new nodes. \end{proof}

\begin{lemma}\label{res-si}
If $W_r$ holds, the leaf node $l$ reached by $Search(e, e)$ satisfies: $W_r \Rightarrow e \in R_l$. \end{lemma}
\begin{proof} Search visiting a node $n$ where $\neg reachable_n (t)$ eventually restarts, so a terminating search only visits reachable nodes in the tree (Lemma \ref{inv-tree}). Search of reachable nodes when $W_r$ holds is regular $k$-ary tree search.\end{proof}

\begin{lemma}\label{res-sl}
If $W_r$ holds, searching the leaf node $l$ from $t_{l1}$ to $t_{l2}$ must read the keys $O(t_{l1}, t_{l2}) \cap R_l$. \end{lemma}
\begin{proof} $l$ is read after a memory barrier, ensuring that $O(t_{l1}, t_{l2}) \cap R_l$ are read.\end{proof}

\begin{lemma}\label{inv-Wr}
All writes to the tree maintain $W_r$. Formally: \begin{center} $\forall t_1, t_2 \in T . (t_1 \le t_2 \wedge W_r (t_1)) \Rightarrow W_r (t_2)$ \end{center}\end{lemma}
\begin{proof} Writes to the tree can be classified into: key insertion, key removal, and rebalancing. 
Rebalancing maintains $W_r$ (Lemma \ref{lrange}).
Key removal and insertion only affects the keys in the tree.
$remove(e_1, e_2, t_1, t_2)$ removes an key from a leaf node $l$, which maintain $W_r$. 
$insert(e, t_1, t_2)$ inserts into leaf nodes for which $\forall t \in T. W_r (t) \Rightarrow e \in R_l$ (Lemma \ref{res-si}), which maintain $W_r$. \end{proof}

\begin{theorem}\label{lost}
ELB-trees are leaf-oriented search trees. \end{theorem}
\begin{proof} ELB-trees are trees and $W_r$ holds initially. All operation on ELB-trees maintains the tree property (Lemma \ref{inv-tree}) and $W_r$ (Lemma \ref{inv-Wr}).\end{proof}

This section proves that ELB-trees are leaf-oriented search trees. Such proofs are sufficient to derive the semantics of concurrent searches and serial insertions and removals. The next section will derive the semantics of the concurrent operations, which requires a few additional lemmas.

\section{Correctness}
\label{sec-4}

This section derives the semantics of the operations. But first we will introduce some terms to reason about the results of such operations. Let: 
\begin{description}
\item[$search(e_1, e_2, t_1, t_2)$] be the result of a search operation matching against keys $e \in [e_1;e_2]$ starting at $t_1$ and ending at $t_2$;
\item[$remove(e_1, e_2, t_1, t_2)$] be the result of a remove operation matching against keys $e \in [e_1;e_2]$ starting at $t_1$ and ending at $t_2$;
\item[$insert(e, t_1, t_2)$] be an insert $e$ operation starting at $t_1$ and ending at $t_2$;
\item[$O(t_1, t_2)$] be the keys that were in $E_r$ at all times during $[t_1;t_2)$: \begin{center} $O(t_1, t_2) = \left\{ e |  \forall t \in [t_1;t_2) . e \in E_r (t)  \right\}$; and \end{center}
\item[$U(t_1, t_2)$] be the keys that were in $E_r$ at some time during $[t_1;t_2)$: \begin{center} $U(t_1, t_2) = \left\{ e |  \exists t \in [t_1;t_2). e \in E_r (t)  \right\}$. \end{center}
\end{description}

We first prove properties of search operations, then derive the operations' semantics:

\begin{lemma}\label{sl}
Searching a set of leaf nodes $RL$ from $t_1$ to $t_2$ reads the keys $\bigcup _{l \in RL} R_l \cap O(t_1, t_2)$. \end{lemma}
\begin{proof}
The search reads the keys $\bigcup _ {l \in RL} R_l \cap O(t_{l1}, t_{l2})$ (Lemma \ref{res-sl}). $\forall l \in RL . O(t_{l1}, t_{l2}) \subseteq O(t_1, t_2)$ holds, as any key in the tree during $t_1$ to $t_2$ must have been in the tree for all fragments of that duration.
\end{proof}

\begin{theorem} $search(e_1, e_2, t_1, t_2)$ can only return $0$ (fail) if there are no matching entries in $E_r$ at all times during $[t_1, t_2)$: \begin{center} $search(e_1, e_2, t_1, t_2) = 0 \Rightarrow [e_1; e_2] \cap O(t_1, t_2) = \emptyset$ \end{center} \end{theorem}
\begin{proof}
$search(e_1, e_2, t_1, t_2) = 0$ implies that a set of leaf nodes $RL$ have been searched, where $[e_1;e_2] \subseteq \bigcup _{l \in RL} R_l$. If there was an key in $[e_1;e_2] \cap O(t_1, t_2)$ it would have been read (Theorem \ref{lost}, Lemma \ref{sl}).\end{proof}

\begin{theorem} Successful searches return a matching key that was in $E_r$ at some point in time during $[t_1;t_2)$: \begin{center}$e = search(e_1, e_2, t_1, t_2) \Rightarrow (e \in U(t_1, t_2) \wedge e \in [e_1; e_2])$\end{center} \end{theorem}
\begin{proof} Successful searches return a key $e$ that was read from a leaf. Since $e$ was read it must have been in $E_r$ (Lemma \ref{sl}).\end{proof}

\begin{theorem} Remove can only return $0$ (fail) if there are no matching entries in $E_r$ at all times during $[t_1, t_2)$:
\begin{center} $remove(e_1, e_2, t_1, t_2) = 0 \Rightarrow O(t_1, t_2) \cap [e_1 ; e_2] = \emptyset$. \end{center}\end{theorem}
\begin{proof} Terminating remove operations that return $0$ have searched a set of leafs $RL$ satisfying $[e_1; e_2] \subseteq \bigcup _{l \in RL} R_l$ (Lemma \ref{sl}), so any keys in $O(t_1, t_2) \cup [e_1;e_2]$ would have been read.\end{proof}

\begin{theorem} Successful remove operations remove matching a key $e$ from $E_r$ that was in $E_r$ at some point in time during $[t_1;t_2)$: \begin{center} $e = remove(e_1, e_2, t_1, t_2) \ne 0 \Rightarrow$ \ $(e_1 \le e \le min(O(t_1, t_2) \cap [e_1 ; e_2]) \le e_2 \wedge e \in U(t_1, t_2))$ \end{center} \end{theorem}
\begin{proof} Terminating remove operations have searched a set of leafs $RL$ satisfying $[e_1; e] \subseteq \bigcup _{l \in RL} R_l$ (Lemma \ref{sl}). Any keys smaller than $e$ in $O(t_1, t_2) \cup [e_1;e_2]$ would have been read.\end{proof}

\begin{theorem} $insert(e, t_1 , t_2)$ adds $e$ to the $E_r$, if $e \notin U (t_1 , t_2 )$. \end{theorem}
\begin{proof} Insert operations terminate when they use a successful CAS operation to write the key into an empty key of a leaf node $l$ where $e \in R_l$ (Lemma \ref{res-si}). The CAS operations success implies the key is not read-only, and hence $reachable_l (t_2)$.\end{proof}

Theorem 2-6 can be summarized as: \\
$e = search( e_1, e_2, t_1, t_2 ) \Rightarrow$ \begin{math} \left\{
     \begin{array}{lr}
       O(t_1, t_2) \cap [e_1 ; e_2] = \emptyset & : e = 0 \\
       e_1 \le e \le e_2 \wedge e \in U(t_1, t_2) & : e \neq 0
     \end{array}
   \right. \end{math}
\\$e = remove(e_1, e_2, t_1, t_2) \Rightarrow$ \begin{math} \left\{
     \begin{array}{lr}
        O(t_1, t_2) \cap [e_1 ; e_2] = \emptyset & : e = 0 \\
       e_1 \le e \le min([e_1; e_2] \cap O(t_1, t_2)) \\ ~ \wedge e \in U(t_1, t_2)  & \raisebox{11pt}{$: e \neq 0$}
     \end{array}
   \right. \end{math}
\ $insert(e, t_1, t_2)$ adds $e$ to $E_r$, if $e \notin U(t_1, t_2)$.

\section{Lock-freedom}
\label{sec-5}

Lock-freedom guarantees that as long as some thread is working on an operation $o_1$, some operation $o_2$ is coming closer to terminating. In this case we say $o_1$ is causing progress, and $o_2$ is making progress. The operations $o_1$ and $o_2$ can be different.
For ELB-trees, this means that whenever a thread is searching, inserting, or removing, some thread must be making progress. 
The following is proof that the operations are lock-free:

\begin{lemma}\label{ter-op}
Operations eventually terminate or restart part of their operation. \end{lemma}
\begin{proof} The operations' algorithms have loops in the following for: node search, tree search, rebalancing, and updating keys in leafs. The algorithms are given in the paper~\cite{bkp13}. Without concurrency, they iterate up to K, tree height, tree height, and 1 times. With concurrency, tree search, rebalancing, and key update loops may restart part of their operation.\end{proof}

\begin{lemma}\label{lf-rebl}
Rebalancing leaf nodes cause progress. \end{lemma}
\begin{proof} If the nodes are written to between deciding to rebalance and rebalancing,  some operation has made progress.
If there are no writes, the size of the first node is either D or S, resulting in balanced nodes of $size \in [min(2 S, 0.5 D); D-1]$. Such nodes can be removed from and inserted into at least once before requiring additional rebalancing. As such, every time a rebalancing completes, one operation has made progress.\end{proof}

\begin{lemma}\label{lf-rebi}
Rebalancing internal nodes cause progress. \end{lemma}
\begin{proof} Rebalancing internal nodes leads to child nodes that can be rebalanced at least one. Each leaf rebalancing cause progress (Lemma \ref{lf-rebl}), hence each internal rebalancing cause progress. \end{proof}

\begin{theorem}\label{lf-s}
Search causes progress. \end{theorem}
\begin{proof} Search eventually terminates, similar to $k$-ary tree search, or rebalances a node (Lemma \ref{res-si}). In the first case the search operation is making progress. In the second case some operation is making progress (Lemma \ref{lf-rebl}, Lemma \ref{lf-rebi}).\end{proof}

\begin{theorem}\label{lf-ri}
Remove and insert operations cause progress. \end{theorem}
\begin{proof} The operations proceed as searches followed by writes to leaf nodes. The leaf node write takes a bounded number of steps, as each key may be read once, but the steps can be restarted due to rebalancing, or other insertions and removals terminating. In the first case, some operation is nearing termination, and in the second case some operation terminated (Lemma \ref{lf-rebl}, Lemma \ref{lf-rebi}).\end{proof}

\section{Conclusion}
\label{sec-6}

This technical report has introduced, proved, and derived properties of ELB-trees. ELB-trees have been proven to be leaf-oriented search trees. Their operations' semantics have been derived as:\\
$e = search( e_1, e_2, t_1, t_2 ) \Rightarrow$ \begin{math} \left\{
     \begin{array}{lr}
       O(t_1, t_2) \cap [e_1 ; e_2] = \emptyset & : e = 0 \\
       e_1 \le e \le e_2 \wedge e \in U(t_1, t_2) & : e \neq 0
     \end{array}
   \right. \end{math}
\\$e = remove(e_1, e_2, t_1, t_2) \Rightarrow$ \begin{math} \left\{
     \begin{array}{lr}
        O(t_1, t_2) \cap [e_1 ; e_2] = \emptyset & : e = 0 \\
       e_1 \le e \le min([e_1; e_2] \cap O(t_1, t_2)) \\ ~ \wedge e \in U(t_1, t_2)  & \raisebox{11pt}{$: e \neq 0$}
     \end{array}
   \right. \end{math}
\ $insert(e, t_1, t_2)$ adds $e$ to $E_r$, if $e \notin U(t_1, t_2)$.
Finally the operations have been proven to be lock-free.

\end{document}